\documentclass[runningheads]{llncs}
\usepackage[T1]{fontenc}

\usepackage{amsfonts}
\usepackage{amsmath}
\usepackage{amssymb}
\DeclareMathOperator*{\argmin}{arg\,min}
\usepackage[ruled,linesnumbered]{algorithm2e}
\usepackage{todonotes}
\usepackage{xparse}
\usepackage{tikz}
\usetikzlibrary{calc}

\SetKwInput{KwInput}{Input}                
\SetKwInput{KwOutput}{Output} 

\def\gyori{Gy\H{o}ri\xspace}
\def\lovasz{Lov\'{a}sz\xspace}
\spnewtheorem*{gl-theorem}{Gy\H{o}ri-Lov\'{a}sz Theorem}{\bfseries}{}

\def\glp{GL-Partition\xspace}
\def\glps{GL-Partitions\xspace}
\def\bigO{\mathcal{O}}
\def\wmax{w_{\mathit{max}}}
\def\compi{\overline{I}}
\def\C{\mathcal{C}}
\def\ourclass{$\text{HHI}_4^2$-free\xspace}

\def\algChordalGL{\texttt{ChordalGL}\xspace}
\def\algChordalGLweighted{\texttt{WeightedChordalGL}\xspace}
\usepackage{hyperref}
\usepackage{cleveref}

\usepackage{graphicx}
\usepackage{color}

\begin{document}
\title{Efficient Constructions for the Gy\H{o}ri-Lov\'{a}sz Theorem on Almost Chordal Graphs}
\titlerunning{Gy\H{o}ri-Lov\'{a}sz Theorem on Almost Chordal Graphs}
%
\author{Katrin Casel\inst{1}\orcidID{0000-0001-6146-8684} \and
Tobias Friedrich\inst{1}\orcidID{0000-0003-0076-6308} \and
Davis Issac\inst{1}\orcidID{0000-0001-5559-7471} \and
Aikaterini Niklanovits\inst{1}\orcidID{0000-0002-4911-4493} \and
Ziena Zeif\inst{1}\orcidID{0000-0003-0378-1458}}
\authorrunning{K. Casel et al.}
%
\institute{Hasso Plattner Institute, University of Potsdam, Potsdam 14482, Germany\\
	\email{\{First name.Last name\}@hpi.de}} 
%
\maketitle              
\begin{abstract}
In the 1970s, \gyori and \lovasz showed that for a $k$-connected $n$-vertex graph, a given set of terminal vertices $t_1, \dots, t_k$ and natural numbers $n_1, \dots, n_k$ satisfying $\sum_{i=1}^{k} n_i = n$, a connected vertex partition $S_1, \dots, S_k$ satisfying $t_i \in S_i$ and $|S_i| = n_i$ exists.
However, polynomial algorithms to actually compute such partitions are known so far only for $k \leq 4$. This motivates us to take a new approach and constrain this problem to particular graph classes instead of restricting the values of $k$.
More precisely, we consider $k$-connected chordal graphs and a broader class of graphs related to them.  
For the first class, we give an algorithm with $\bigO(n^2)$ running time that solves the problem exactly, and for the second, an algorithm with $\bigO(n^4)$ running time that deviates on at most one vertex from the required vertex partition sizes.

\end{abstract}
\section{Introduction}

Partitioning a graph into connected subgraphs is a fundamental task in graph algorithms. Such \emph{connected} partitions occur as desirable structures in many application areas such as  image processing~\cite{lucertini1993most}, road network decomposition~\cite{mohring2007partitioning}, and robotics \cite{zhou2019balanced}.

From a theoretical point of view, the existence of a partition into connected components with certain properties also gives insights into the graph structure.
 In theory as well as in many applications, one is interested in a connected partition that has a given number of subgraphs of chosen respective sizes. With the simple example of a star-graph, it is observed that not every graph admits a connected partition for any such choice of subgraph sizes. More generally speaking, if there exists a small set of $t$ vertices whose removal disconnects a graph (\emph{separator}), then any connected partition into $k>t$ subgraphs has limited choice of subgraph sizes. Graphs that do not contain  such a separator of size less than $k$ are called $k$-\emph{connected}.

On the other hand, \gyori and \lovasz independently showed that $k$-connectivity is not just necessary but also sufficient to enable a connected partitioning into $k$ subgraphs of required sizes, formally stated by the following result.
	
\begin{gl-theorem}[\cite{gyori1976division},\cite{lovasz1977homology}]
	Let $k\geq 2$ be an integer, $G=(V,E)$ a $k$-connected graph, \ $t_1,\ldots, t_k\in V$ distinct vertices and $n_1,\ldots,n_k\in \mathbb N$ such that $\sum_{i=1}^kn_i= |V|$. Then $G$ has disjoint connected subgraphs $G_1, \ldots G_k$ such that $|V(G_i)|=n_i$  and $t_i\in V(G_i)$ for all $i\in [k]$.
\end{gl-theorem}

The caveat of this famous result is that the constructive proof of it yields an exponential time algorithm. Despite this result being known since 1976, to this day we only know polynomial constructions for restricted values of~$k$. Specifically, in 1990 Suzuki et al.~\cite{suzuki1990linear} provided such an algorithm for $k=2$ and also for $k=3$~\cite{suzuki1990algorithm}. Moreover in 1994 Wada et al.~\cite{wada1993efficient} also provided an extended result for $k=3$.
For the case of $k=4$ Nakano et al.~\cite{nakano1997linear} gave a linear time algorithm for the case where $k=4$, $G$~is planar and the given terminals are located on the same face of a plane embedding of~$G$, while in 2016 Hoyer and Thomas~\cite{hoyer2019independent} provided a polynomial time algorithm for the case of $k=4$. And so far, this is where the list ends, thus for $k\geq 5$ it remains open whether there even exists a polynomial time construction.

Towards a construction for general $k$, we consider restricting the class of $k$-connected graphs instead of the values of $k$. 
More precisely, we consider (generalizations of) \emph{chordal} $k$-connected graphs.
A graph is called \emph{chordal}, if it does not contain an induced cycle of length more than three. 
The restriction to chordal graphs is known to often yield tractability for otherwise NP-hard problems, for example chromatic number, clique number, independence number, clique covering number, stable set and treewidth decomposition \cite{rose1976algorithmic}. 
Apart from the interest chordal graphs have from a graph theoretic point of view, their structural properties have also been proven useful in biology when it comes to studying multidomain proteins and network motifs (see e.g.~\cite{przytycka2005graph,przytycka2006important}).

\paragraph*{Our contribution}
To the best of our knowledge, this paper is the first to pursue the route of restricting the \gyori-\lovasz Theorem to special graph classes in order to develop a polynomial construction for general values of $k$ on a non-trivial subclass of $k$-connected graphs.
We believe that in general considering the structure of the minimal separators of a graph is promising when it comes to developing efficient algorithms for  the \gyori-\lovasz Theorem.

We give a constructive version of the \gyori-\lovasz Theorem for \emph{chordal} $k$-connected graphs with a running time in $\bigO(|V|^2)$. Observe here that this construction works for all values of $k$.
Then we show how this result can be generalized in two directions.

First, we generalize our result to the vertex weighted version of the \gyori-\lovasz Theorem (as proven independently by Chandran et al.~\cite{chandran2018spanning}, Chen et al.~\cite{chen2007almost} and Hoyer~\cite{hoyer2019independent}), specifically deriving the following theorem.
\begin{theorem}	\label{thm::GL-theorem-chordal-weighted}
	Let $k\geq 2$ be an integer, $G=(V,E,w)$ a vertex-weighted $k$-connected chordal graph with $w \colon V \to \mathbb{N}$, $t_1,\ldots, t_k \in V$ distinct vertices, and $w_1,\ldots,w_k \in\mathbb N$ with $w_i \geq w(t_i)$ and $\sum_{i=1}^k w_i = w(V)$ for all $i \in [k]$.
A partition $S_1,\dots,S_k$ of $V$, such that  $G[S_i]$ is connected, $t_i \in S_i$ and $w_i - \wmax < w(S_i) < w_i + \wmax$, for all $i \in [k]$, can be computed in time $\bigO(|V|^2)$.
\end{theorem}

We further use this weighted version to derive an approximate version of the \gyori-\lovasz Theorem for a larger graph class.
Specifically we define $\text{I}_j^i$ to contain all graphs that occur from two distinct chordless $C_j$'s that have at least $i$ vertices in common. We focus on $\text{I}_4^2$-free combined with HH-free graphs.
More specifically, we consider the subclass of $k$-connected graphs that contain no hole or house as subgraph (see preliminaries for the definitions of structures such as hole, house etc.) and that does not contain two distinct induced $C_4$ that share more than one vertex.
We call this class of graphs \ourclass.
Note that \ourclass, apart from being a strict superclass of chordal graphs, is also a subclass of HHD-free graphs (that is house, hole, domino-free graphs), a graph class studied and being used in a similar manner as chordal graphs as it is also a class where the minimum fill-in set is proven to be polynomially time solvable~\cite{DBLP:conf/wg/BroersmaDK97} (see also \cite{jamison1988semi} for NP-hard problems solved in polynomial time on HHD-free graphs).
Taking advantage of the fact that given an \ourclass  graph, the subgraph formed by its induced $C_4$ has a treelike structure, we are able to derive the following result.
\begin{theorem}\label{thm:ourclass-GL-partition}
	Let $k\geq 2$ be an integer, $G=(V,E,w)$ a vertex-weighted $k$-connected \ourclass graph with $w \colon V \to \mathbb{N}$, $t_1,\ldots, t_k \in V$ distinct vertices, and $w_1,\ldots,w_k \in\mathbb N$ with $w_i \geq w(t_i)$ and $\sum_{i=1}^k w_i = w(V)$ for all $i \in [k]$.
A partition $S_1,\dots,S_k$ of $V$, such that  $G[S_i]$ is connected, $t_i \in S_i$ and $w_i - 2\wmax < w(S_i) < w_i + 2\wmax$, for all $i \in [k]$, can be computed in time $\bigO(|V|^4)$.
\end{theorem}
Notice that the above theorem implies a polynomial time algorithm with an additive error of $1$ for the unweighted case.

\section{Preliminaries}

All graphs mentioned in this paper are undirected, finite and simple.
Given a graph $G$ and a vertex $v\in V(G)$ we denote its \emph{open neighborhood} by $N_G(v):=\{u\in V(G) \mid uv\in E(G)\}$ and by $N_G[v]$ its \emph{closed neighborhood}, which is $N(v)\cup \{v\}$. 
Similarly we denote by $N_G(S):=\bigcup_{v\in S}N_G(v)\setminus S$ the open neighborhood of a vertex set $S\subseteq V(G)$ and by $N_G[S]:=N_G(S)\cup S$ its closed neighborhood.
We omit the subscript $G$ when the graph we refer to is clear from the context. A vertex $v\in V(G)$ is \emph{universal to a vertex set $S \subset V(G)$} if $S\subseteq N(v)$. 
Let $G$ be a graph and $S\subseteq V(G)$. The \emph{induced subgraph} from $S$, denoted by $G[S]$, is the graph with vertex set $S$ and all edges of $E(G)$ with both endpoints in $S$.

A graph $G$ is \emph{chordal} if any cycle of $G$ of size at least $4$ has a chord (i.e., an edge linking two non-consecutive vertices of the cycle).
A vertex $v\in V(G)$ is called \emph{simplicial} if $N[v]$ induces a clique.
Based on the existence of simplicial vertices in chordal graphs, the following notion of vertex ordering was given.
Given a graph $G$, an ordering of its vertices $(v_1,\ldots,v_n)$ is called \emph{perfect elimination ordering} (p.e.o.) if $v_i$ is  simplicial in $G[\{v_i,v_{i+1},\ldots, v_n\}]$ for all $i\in[n]$. Given such an ordering $\sigma: V(G)\rightarrow\{1,\ldots,n\}$ and a vertex $v\in V(G)$ we call $\sigma(v)$ the \emph{p.e.o.~value of $v$}. 
Rose et al.~\cite{rose1976algorithmic} proved that a p.e.o.~of any chordal graph can be computed in linear time.

Let $e=\{u,v\}$ be an edge of $G$. 
We denote by $G/e$ the graph $G'$, that occurs from $G$ by the contraction of $e$, that is, by removing $u$ and $v$ from $G$ and replacing it by a new vertex $z$ whose neighborhood is $\left(N(u)\cup N(v)\right)\setminus\{u,v\}$. 

A graph $G$ is \emph{connected} if there exists a path between any pair of distinct vertices.
Moreover, a graph is \emph{$k$-connected} for some $k\in\mathbb N$ if after the removal of any set of at most $k-1$ distinct vertices $G$ remains connected.
Given a graph $G$ and a vertex set $S \subseteq V(G)$, we say that $S$ is a \emph{separator} of $G$ if its removal disconnects $G$.
We call $S$ a \emph{minimal separator} of $G$ if the removal of any subset $S' \subseteq V(G)$ with $|S'|<|S|$ results in a connected graph.

We now define some useful subgraphs, see also Figure~\ref{fig:: subgraphs} for illustrations.
An induced chordless cycle of length at least $5$ is called a \emph{hole}.
The graph that occurs from an induced chordless $C_4$ where exactly two of its adjacent vertices have a common neighbor is called a \emph{house}.
When referring to just the induced $C_3$ that is part of a house we call it \emph{roof} while the induced $C_4$ is called \emph{body}.
Two induced $C_4$ sharing exactly one edge form a \emph{domino}.
A graph that contains no hole, house or domino as an induced subgraph is called HHD-free.
We call a graph that consists of two $C_4$ sharing a vertex, and an edge that connects the two neighbors of the common vertex in a way that no other $C_4$ exists a \emph{double house}.

Lastly, let $G=(V,E)$ be a $k$-connected graph, let $t_1,\ldots, t_k\in V$ be $k$ distinct vertices, and let $n_1,\ldots, n_k$ be natural numbers satisfying $\sum_{i=1}^kn_i= |V|$.
We call $S_1,\ldots S_k \subseteq V(G)$ a \emph{GL-Partition of $G$} if $S_1,\ldots S_k$ forms a partition of $V(G)$, such that for all $i \in [k]$ we have that $G[S_i]$ is connected, $t_i \in S_i$ and $|S_i|=n_i$.
When there exists an $l\in\mathbb N$, such that for such a partition only $n_i-l\leq|S_i|\leq n_i+l$ holds instead of $|S_i|=n_i$, we say that $S_1,\ldots,S_k$ is a \emph{GL-Partition of $G$ with deviation $l$}.

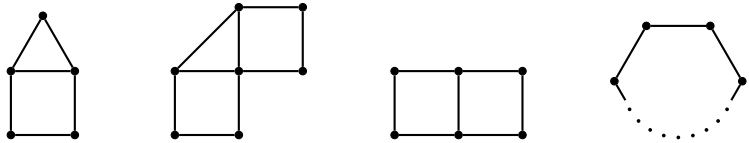
\begin{figure}
	\centering
	\begin{tikzpicture}[%
		scale=0.85,
		vertex/.style={draw,fill,circle,inner sep=1pt},%
		edge/.style={draw,thick}]
		\node[vertex] (v00) {};
		\node[vertex] (v10) at ($(v00)+(1,0)$) {};
		\node[vertex] (v01) at ($(v00)+(0,1)$) {};
		\node[vertex] (v11) at ($(v10)+(0,1)$) {};
		\node[vertex] (vtop) at ($(v11)+(120:1)$) {};
		\foreach \u/\v in {v00/v10,v10/v11,v11/v01,v01/v00,v11/vtop,vtop/v01}
		\draw[edge] (\u)--(\v);
	\end{tikzpicture} \hspace{1cm}
	\begin{tikzpicture}[%
		scale=0.85,
		vertex/.style={draw,fill,circle,inner sep=1pt},%
		edge/.style={draw,thick}]
		\node[vertex] (v00) {};
		\node[vertex] (v10) at ($(v00)+(1,0)$) {};
		\node[vertex] (v01) at ($(v00)+(0,1)$) {};
		\node[vertex] (v11) at ($(v10)+(0,1)$) {};
		\node[vertex] (v21) at ($(v11)+(1,0)$) {};
		\node[vertex] (v12) at ($(v11)+(0,1)$) {};
		\node[vertex] (v22) at ($(v21)+(0,1)$) {};
		\foreach \u/\v in {v00/v10,v10/v11,v11/v01,v01/v00,v11/v21,v21/v22,v22/v12,v12/v11,v12/v01}
		\draw[edge] (\u)--(\v);
	\end{tikzpicture}\hspace{1cm}
	\begin{tikzpicture}[%
		scale=0.85,
		vertex/.style={draw,fill,circle,inner sep=1pt},%
		edge/.style={draw,thick}]
		\node[vertex] (v00) {};
		\node[vertex] (v10) at ($(v00)+(1,0)$) {};
		\node[vertex] (v01) at ($(v00)+(0,1)$) {};
		\node[vertex] (v11) at ($(v10)+(0,1)$) {};
		\node[vertex] (v20) at ($(v10)+(1,0)$) {};
		\node[vertex] (v21) at ($(v11)+(1,0)$) {};
		\foreach \u/\v in {v00/v10,v10/v11,v11/v01,v01/v00,v10/v20,v20/v21,v21/v11}
		\draw[edge] (\u)--(\v);
	\end{tikzpicture}\hspace{1cm}
	\begin{tikzpicture}[%
		scale=0.85,
		vertex/.style={draw,fill,circle,inner sep=1pt},%
		edge/.style={draw,thick},%
		dot/.style={draw,fill,circle,inner sep=0.1pt}]
		\coordinate (O);
		\foreach \i in {0,1,2,3,4,5}
		\coordinate (C\i) at ($(O)+(360*\i/6:1)$);
		\foreach \i in {14,...,22}
		\node[dot] at ($(O)+(360*\i/24:0.88)$) {};
		\node[vertex] (v0) at (C0) {};
		\node[vertex] (v1) at (C1) {};
		\node[vertex] (v2) at (C2) {};
		\node[vertex] (v3) at (C3) {};
		\node (v35) at ($0.5*(C3)+0.5*(C4)$) {};
		\node (v55) at ($0.5*(C5)+0.5*(C0)$) {};
		\foreach \u/\v in {v0/v1,v1/v2,v2/v3,v3/v35,v55/v0}
		\draw[edge] (\u)--(\v);
	\end{tikzpicture}
	\caption{Specific subgraphs used throughout the paper, from left to right: house, double house, domino and hole example}\label{fig:: subgraphs}
\end{figure}

\section{\glp for Chordal Graphs}

We present a simple implementable algorithm with quadratic running time that computes \glps in chordal graphs. We then show that a slight modification of our algorithm is sufficient to compute a \glp on a vertex weighted graph,  thus proving \Cref{thm::GL-theorem-chordal-weighted}.

Due to space restrictions the proof of Lemma \ref{lemma::ordering_induced_path} has been moved to Appedix~\ref{appendix2}, and the proof of correctness of Algorithm~\ref{alg::chordal-weighted} to Appendix~\ref{appendix1}.

\subsection{\glp for Unweighted Chordal Graphs}
For simplicity, we first prove the restricted version of \Cref{thm::GL-theorem-chordal-weighted} to unweighted graphs. 
We use a p.e.o.~to compute a vertex partition, as described formally in Algorithm~\ref{alg::chordal}.
This algorithm receives as input a $k$-connected chordal graph $G=(V,E)$, terminal vertices $t_1,\ldots,t_k \in V$, and natural numbers $n_1,\ldots,n_k$ satisfying $\sum_{i=1}^kn_i=n$,
and outputs connected vertex sets $S_1,\ldots,S_k \subseteq V$ such that $|S_i|=n_i$ and $t_i \in S_i$.
In the beginning of the algorithm we initialize each set $S_i$ to contain only the corresponding terminal vertex $t_i$, and add vertices iteratively to the \textit{non-full sets} ($S_i$'s that have not reached their demanded size).
We say a vertex \textit{$v$ is assigned} if it is already part of some $S_i$ and \textit{unassigned} otherwise.
At each iteration, the unassigned neighborhood of the union of the previously non-full sets is considered, and the vertex with the minimum p.e.o.~value is selected to be added to a non-full set.
In case there is more than one non-full set in the neighborhood of this vertex, it is added to the one with lowest priority, where the priority of each set is defined to be the largest p.e.o.~value of its vertices so far.
The algorithm terminates once all vertices are assigned, in $\bigO({|V|}^2)$ time.


\begin{algorithm}[h]
	\KwInput{$k$-connected chordal graph $G=(V,E)$, terminal vertices $t_1,\ldots,t_k \in V$, and natural numbers $n_1,\ldots,n_k$ satisfying $\sum_{i=1}^kn_i=n$}
	\KwOutput{Connected vertex sets $S_1,\ldots,S_k \subseteq V$ such that $|S_i|=n_i$ and $t_i\in S_i$}
	
	$\sigma$ $\gets$ Compute p.e.o.~of $G$ as function $\sigma \colon V \to |V| $\\
		
	$S_i\leftarrow \{t_i$\}, for all $i \in [k]$
	
	\While{$\bigcup_{i \in [k]} S_i\neq V(G)$}{
		
		$I \gets \{i \in [k] \mid |S_i| < n_i\}$
		
		$V' \gets N(\bigcup_{i \in I} S_i) \setminus \bigcup_{i \in [k]\setminus I} S_i$
		
		$v' \gets \argmin_{v\in V'}\sigma(v)$
		
		$J \gets \left\{i \in I \mid v \in N(S_i) \right\}$
		
		$j' \gets$ $\argmin_{j\in J}\max(\sigma(S_{j}))$
		
		$S_{j'} \gets S_{j'} \cup \{v'\}$
	}

	\Return $S_1,\dots,S_k$
	\caption{\algChordalGL}
	\label{alg::chordal}
\end{algorithm}

For the correctness of Algorithm~\ref{alg::chordal} it is enough to show that the unassigned neighborhood $V'$ of all non-full sets is not empty in each iteration of the while-loop, since this implies that we enlarge a non-full set (in the algorithm denoted as $S_{j'}$) by one vertex (in the algorithm denoted as $v'$) while maintaining the size of all remaining sets.
That is, in each iteration we make progress in the sense that $|\bigcup_{i \in [k]} S_i|$ increases while maintaining the invariant $|S_i| \leq n_i$ for all $S_i$'s.
Note that $v' \in N(S_{j'})$ which in turn implies that $G[S_i]$ is always connected for all $i \in [k]$.
Finally, by $\sum_{i=1}^k n_i=n$ and through the way we update $I$ we ensure that the algorithm (or while-loop) terminates as $\bigcup_{i \in [k]} S_i = V$ only if we have $|S_i| = n_i$ for all $S_i$'s.

Towards proving the required lemmata for the correctness of Algorithm~\ref{alg::chordal} we make the following observation for the p.e.o.~of a graph.

\begin{lemma}
	\label{lemma::ordering_induced_path}
	Let $\sigma$ be a p.e.o~of a graph $G=(V,E)$ and $P=\{v_1,v_2,\ldots,v_k\}$ a vertex set of $G$ that induces a simple path with endpoints $v_1$ and $v_k$. 
	Then $\sigma(v_i)>\min\{\sigma(v_1),\sigma(v_k)\}$ for all $i=2,\ldots,k-1$.
\end{lemma}

\begin{lemma}
	\label{lemma::neighborhood_non_empty}
	In each iteration of the while-loop in Algorithm~\ref{alg::chordal} we have $V' \neq \varnothing$.
\end{lemma}
\begin{proof}
	We first define the \textit{$z$-connecting} neighborhood of a vertex $v$ to be the neighbors of $v$ that are included in some induced path connecting $v$ to $z$.
	
	We prove that every non-full set $S_i$ contains a vertex in its neighborhood $N(S_i)$ that is unassigned, which implies that $V' \neq \varnothing$.
	Assume for a contradiction that at some iteration of our algorithm there is an non-full set $S_i$ whose neighborhood is already assigned to other sets.
	Let $v$ be the vertex of $S_i$ of maximum $\sigma$ value among its vertices and $z$ be the vertex of maximum $\sigma$ value among the unassigned vertices. 
	Note that $vz\not \in E(G)$.
	Let $\mathcal P$ be the set of all simple induced paths of $G$ with endpoints $z$ and $v$. Consider now the following cases:
	\begin{enumerate}
		\item 
		If $\sigma(z)>\sigma(v)$, we get from \Cref{lemma::ordering_induced_path} that every internal vertex of each path in $\mathcal P$ has higher $\sigma$ value than $v$.
		Note that no vertex of $S_i$ is an internal vertex of some path in $\mathcal P$, since all of them have smaller $\sigma$ value than $v$ by the selection of $v$.
		Denote the $z$-connecting neighborhood of $v$ by $C$.
		
		Let $a,b$ be two vertices in $C$ and assume that $a,b\in S_j$ for some $j$.
		Assume also that during our algorithm, $a$ is added to $S_j$ before $b$.
		Since all vertices of $S_i$ have smaller $\sigma$ value than both $a$ and $b$, and $a$ is added to $S_j$ before $b$, the moment $b$ is added to $S_j$, $S_i$ has already been formed.
		Consider now the iteration that this happens.
		Since $b\in N(v)$, $G[S_i\cup \{b\}]$ is connected.
		Moreover since $\sigma (a)>\sigma(v)$ and $S_i$ is not full, $b$ should be added to $S_i$ instead of $S_j$.
		As a result each set apart from $S_i$ contains at most one such neighbor of $v$, and hence $|C|<k$.
		
		Observe that $G\setminus C$ has no induced path connecting $z$ and $v$ which in turn implies that $G\setminus C$ has no $z-v$ path in general.
		However, this contradicts the $k$-connectivity of $G$.
		
		\item 
		If $\sigma(z)<\sigma(v)$,
		since $z$ is the unassigned vertex of the highest $\sigma$ value among all unassigned vertices, and by \Cref{lemma::ordering_induced_path} all vertices in $\mathcal P$ have greater $\sigma$ value than $z$, all of its $v$-connecting neighbors in $\mathcal P$ are already assigned in some set.
		Denote the set of $v$-connecting neighbors of $z$ by $C$.
		
		Assume now that there are two vertices of $C$, $a$ and $b$, that are contained in some $S_j$ and assume also without loss of generality that $a$ was added to $S_j$ before $b$.
		Note that since $\sigma(z)<\sigma(b)$ at each iteration of our algorithm $z$ is considered before $b$ to be added to some set if the induced graph remains connected.
		As a result, after $a$ is added to $S_j$, the induced subgraph $G[S_j \cup \{z\}]$ is connected and hence $z$ should be added to $S_j$ before $b$.
		
		This means that each set contains at most one $v$-connecting neighbor of $z$ and therefore $|C|<k$.
		Since $G\setminus C$ has no induced path connecting $z$ and $v$, there is no $z$-$v$-path in $G \setminus C$, which contradicts the $k$-connectivity. 
	\end{enumerate}
\end{proof}



\begin{corollary}
	\label{corollary::never_iso}
	At each iteration of Algorithm~\ref{alg::chordal}, unless all vertices are assigned, the neighborhood of each non-full set contains at least one unassigned vertex.  
\end{corollary}
In the weighted case we use the above corollary of \Cref{lemma::neighborhood_non_empty}.
In particular, it follows from \Cref{corollary::never_iso} that as long as we do not declare a set to be full, we ensure that we are able to extend it by a vertex in its neighborhood that is unassigned.
Note that in the weighted case we do not know in advance how many vertices are in each part.

\subsection{\glp for Weighted Chordal Graphs}

With a slight modification of Algorithm~\ref{alg::chordal} we can compute the weighted version of a \glp.
In particular, we prove \Cref{thm::GL-theorem-chordal-weighted}.

The input of our algorithm differs from the unweighted case by having a positive vertex-weighted graph $G=(V,E,w)$ and instead of demanded sizes $n_1, \dots, n_k$ we have demanded weights $w_1, \dots, w_k$ for our desired vertex sets $S_1, \dots, S_k$, where $\sum_{i=1}^{k} w_i = w(V)$.
Note also that  $w(S_i)$ is not allowed to deviate more than $\wmax = \max_{v \in V} w(v)$ from $w_i$, i.e.~$w_i - \wmax < w(S_i) < w_i + \wmax$.
\begin{algorithm}[h]
	\KwInput{$k$-connected vertex-weighted chordal graph $G(V,E,w)$, terminal vertices $t_1,\ldots,t_k \in V$, and positive weights $w_1,\ldots,w_k$ satisfying $\sum_{i=1}^k w_i=w(V)$}
	\KwOutput{Connected vertex sets $S_1,\ldots,S_k \subseteq V$ such that $w_i - w_{\max} < w(S_i) < w_i + w_{\max}$ and $t_i\in S_i$}
	
	$\sigma$ $\gets$ Compute p.e.o.~of $G$ as function $\sigma \colon V \to |V| $\\
	
	$S_i\leftarrow \{t_i$\}, for all $i \in [k]$
	
	$I \gets \{i \in [k] \mid w(S_i) < w_i\}$
	
	\While{$|I| \neq 1$ \textbf{and} $\bigcup_{i \in [k]} S_i\neq V(G)$}{
		
		
		$V' \gets N(\bigcup_{i \in I} S_i) \setminus \bigcup_{i \in [k] \setminus I} S_i$
		
		$v' \gets \argmin_{v\in V'}\sigma(v)$
		
		$J \gets \left\{i \in I \mid v \in N(S_i) \right\}$
		
		$j' \gets$ $\argmin_{j\in J}\max(\sigma(S_{j}))$
		
		\If{$w(S_{j'}) + w(v') < w_{j'}$}
		{
			$S_{j'} \gets S_{j'} \cup \{v'\}$
		}
		\Else
		{
			$I \gets I \setminus \{j'\}$
			
			\If{$\sum_{i \in [k] \setminus I} (w_i - w(S_i)) \geq 0$ \textbf{or} $w(S_{j'}) + w(v') = w_{j'}$}
			{
				\label{alg::weighted-if2}
				$S_{j'} \gets S_{j'} \cup \{v'\}$
				
			}
		}
	}
	If $|I| = 1$, assign all vertices $V \setminus \bigcup_{i \in [k]} S_i$ (possibly empty) to $S_j$ with $j \in I$.
	\label{algWeightedChordal::last-step} 
	\caption{\algChordalGLweighted}
	\label{alg::chordal-weighted}
\end{algorithm}

Again we set each terminal vertex $t_i$ to a corresponding set $S_i$, and enlarge iteratively the \textit{non-full weighted sets} ($S_i$'s that are not declared as full).
One difference to the previous algorithm is that we declare a set $S_i$ as \textit{full weighted set}, if together with the next vertex to be potentially added its weight would exceed $w_i$.
After that, we decide whether 
to add the vertex with respect to the currently full weighted sets.
Similar to Algorithm~\ref{alg::chordal} we interrupt the while-loop if $S_1,\dots,S_k$ forms a vertex partition of $V$ and the algorithm terminates.
However, to ensure that we get a vertex partition in every case, we break the while-loop when only one non-full weighted set is left and assign all remaining unassigned vertices to it.


Observe that we can make use of \Cref{corollary::never_iso}, since Algorithm~\ref{alg::chordal-weighted} follows the same priorities concerning the p.e.o.~as Algorithm~\ref{alg::chordal}.
Basically, it implies that as long we do not declare a set as full weighted set and there are still unassigned vertices then those sets have unassigned vertices in its neighborhood.

We conclude this section by extending the above algorithms to graphs having distance $k/2$ from being chordal. 
In particular this corollary is based on the observation that an edge added to a graph does not participate in any of the parts those algorithms output if both of its endpoints are terminal vertices.

\begin{corollary}
	Let $G$ be a $k$-connected graph which becomes chordal after adding $k/2$ edges. Given this set of edges, a \glp (also its weighted version) can be computed in polynomial time but without fixed terminals.
\end{corollary}







\section{\glp for \ourclass}
This section is dedicated to the proof of \Cref{thm:ourclass-GL-partition}. The underlying idea  for this result is to carefully contract edges to turn a $k$-connected \ourclass graph into a  chordal graph that is still $k$-connected. Note that we indeed have to be very careful here to find a set of contractions, as we need it to satisfy three seemingly contradicting properties: removing all induced $C_4$, preserving $k$-connectivity, and contracting at most one edge adjacent to each vertex. The last property is needed to bound the maximum weight of the vertices in the contracted graph. Further, we have to be careful not to contract terminal vertices. 

The computation for the unweighted case of the partition for \Cref{thm:ourclass-GL-partition} is given in Algorithm~\ref{alg:: ourclassGL} below,  which is later extended to the weighted case as well. Note that we can assume that $n_i\geq 2$ since if $n_i=1$ for some $i\in[k]$ we simply declare the terminal vertex to be the required set and remove it from $G$. This gives us a $(k-1)$-connected graph and $k-1$ terminal vertices.


%
%
%
%
%
%
%
%
Before starting to prove the Lemmata required for the correctness of Algorithm~\ref{alg:: ourclassGL} we give a structural insight which is used in almost all proofs of the following Lemmata.
Due to space restrictions the proofs of Lemmata \ref{lemma:observation} to  \ref{lemma:closed-under-contraction}, \ref{lemma:contraction-connectivity} and \ref{lemma:C4-tree} have been moved to Appendix~\ref{appendix2}.
\begin{lemma}\label{lemma:observation}
	Given an \ourclass graph $G$ and an induced $C_4$, $C\subseteq V(G)$, then any vertex in $V(G)\setminus C$ that is adjacent to two vertices of $C$ is universal to $C$. Moreover, the set of vertices that are universal to $C$ induces a clique.
\end{lemma}

\begin{algorithm}[H]
	\KwInput{$k$-connected \ourclass graph $G(V,E)$, terminal vertices $t_1,\ldots,t_k \in V$, and positive integers $n_1,\ldots,n_k \geq 2$ satisfying $\sum_{i=1}^k n_i=n$}
	\KwOutput{Connected vertex sets $S_1,\ldots,S_k \subseteq V$ such that $n_i-1\leq|S_i|\leq n_i+1$ and $t_i\in S_i$}
	
	Add an edge between each pair of non-adjacent terminals that are part of an induced $C_4$
	
	$\C \gets $ Set of all induced $C_4$ in $G$.
	
	$G' \gets$ $(\bigcup_{C \in \mathcal{C}} V(C), \bigcup_{C \in \mathcal{C}} E(C))$
	
	$E' \gets \varnothing$
	
	\While{$\C \neq \varnothing$}
	{
		Select three vertices $v_1,v_2,v_3$ in $G'$ and the corresponding cycle $C \in \C$ that satisfies that for all $C' \in \C \setminus \{C\}$ we have $V(C') \cap \{v_1,v_2,v_3\} = \varnothing$. 
		
		Pick a vertex $v$ from $v_1,v_2,v_3$ that is not a terminal vertex and add an incident edge of $v$ in $G'[\{v_1,v_2,v_3\}]$ to $E'$.
		
		Remove the cycle $C$ from $\C$ and the vertices $v_1,v_2,v_3$ from $G'$.
	}
	
	Transform $G$ to a weighted graph $G''$ by contracting each edge of $E'$ in $G$, assigning to each resulting vertex as weight the number of original vertices it corresponds to. 
	
	$S_1, \dots S_k \gets$ Run Algorithm~\ref{alg::chordal-weighted} with $G''$, the given set of terminals $t_1,\ldots,t_k$, and the size (or weight) demands $n_1,\ldots,n_k$  as input.
	
		Reverse the edge contraction of $E'$ in the sets $S_1, \dots, S_k$ accordingly.
	\caption{\ourclass GL}
\label{alg:: ourclassGL}
\end{algorithm}


\begin{lemma}\label{lemma:doublehouse}
	Let $G$ be an \ourclass graph. If $G$ contains a double house as a subgraph then at least one of the two $C_4$ in it has a chord.
\end{lemma}

\begin{figure}[h]
	\centering
\begin{tikzpicture}[%
	scale=0.85,
	vertex/.style={draw,fill,circle,inner sep=1pt},%
	edge/.style={draw,thick}]
	\node[vertex, label=below left:{$u_{13}$}] (v00) {};
	\node[vertex, label=below right:{$u_{11}$}] (v10) at ($(v00)+(1,0)$) {};
	\node[vertex, label=below left:{$u_{3}$}] (v01) at ($(v00)+(0,1)$) {};
	\node[vertex, label= below right:{$u_1$}] (v11) at ($(v10)+(0,1)$) {};
	\node[vertex, label=below right:{$u_{21}$}] (v21) at ($(v11)+(1,0)$) {};
	\node[vertex, label=above right:{$u_{2}$}] (v12) at ($(v11)+(0,1)$) {};
	\node[vertex, label=above right:{$u_{22}$}] (v22) at ($(v21)+(0,1)$) {};
	\node at ($0.5*(v00)+0.5*(v11)$) {$C^1$};
	\node at ($0.5*(v11)+0.5*(v22)$) {$C^2$};
	\foreach \u/\v in {v00/v10,v10/v11,v11/v01,v01/v00,v11/v21,v21/v22,v22/v12,v12/v11,v12/v01}
	\draw[edge] (\u)--(\v);
\end{tikzpicture}
\begin{tikzpicture}[%
	scale=0.85,
	vertex/.style={draw,fill,circle,inner sep=1pt},%
	edge/.style={draw,thick}]
	\node[vertex, label=below:{$z_2$}] (v00) {};
	\node[vertex, label=above:{$z_1$}] (v10) at ($(v00)+(0.5,0.5)$) {};
	\node[vertex, label=above:{$z_3$}] (v01) at ($(v00)+(0,1)-(0.5,0.5)$) {};
	\node[vertex, label= above:{$z_4$}] (v11) at ($(v00)+(0,1)$) {};
	
	\node[vertex, label=left:{$u$}] (u) at ($0.5*(v00)+0.5*(v11)+(-2,0)$) {};
	\node[vertex, label=right:{$w$}] (w) at ($(v10)+(1.5,0)$) {};
	\node[vertex, label=above:{$u_1$}] (u1) at ($(v01)+(-0.5,0)$) {};
	\node[vertex, label=above:{$u_3$}] (u3) at ($(v01)+(-0.5,1)$) {};
	\node[vertex, label=below:{$u_2$}] (u2) at ($(v00)+(-1,-0.5)$) {};
	\node[vertex, label=below:{$w_2$}] (w2) at ($(v00)+(1,-0.5)$) {};
	\node[vertex, label=above:{$w_3$}] (w3) at ($(v11)+(1,0.5)$) {};
	\node[vertex, label=below:{$w_1$}] (w1) at ($(v10)+(0.5,0)$) {};	
	\draw[edge,dashed] (u)--(u1);
	\draw[edge, dashed] (u)--(u2);
	\draw[edge, dashed] (u)--(u3);
	\draw[edge, dashed] (w)--(w1);
	\draw[edge, dashed] (w)--(w2);
	\draw[edge, dashed] (w)--(w3);
	\draw[edge] (u1)--(v01);
	\draw[edge] (u3)--(v11);	
	\draw[edge] (v11)--(w3);
	\draw[edge] (u2)--(v00);
	\draw[edge] (v00)--(w2);
		\draw[edge] (v10)--(w1);
	\foreach \u/\v in {v00/v10,v10/v11,v11/v01,v01/v00}
	\draw[edge] (\u)--(\v);
\end{tikzpicture}
\begin{tikzpicture}[%
	scale=0.85,
	vertex/.style={draw,fill,circle,inner sep=1pt},%
	edge/.style={draw,thick}]
	\node[vertex, label= left:{$v_2$}] (v00) {};
	\node[vertex, label=right:{$v_3$}] (v10) at ($(v00)+(1,0)$) {};
	\node[vertex, label=left:{$v_1$}] (v01) at ($(v00)+(0,1)$) {};
	\node[vertex, label= right:{$v_4$}] (v11) at ($(v10)+(0,1)$) {};
	\node[vertex, label=left:{$u$}] (u) at ($0.5*(v00)+0.5*(v01)+(-2,0)+(0.5,0)$) {};
	\node[vertex, label=right:{$w$}] (w) at ($0.5*(v11)+0.5*(v10)+(2,0)-(0.5,0)$) {};
	\node[vertex, label=above:{$u_1$}] (u1) at ($(v01)+(-0.5,0.5)$) {};
	\node[vertex, label=below:{$u_2$}] (u2) at ($(v00)+(-0.5,-0.5)$) {};
	\node[vertex, label=above:{$w_1$}] (w1) at ($(v11)+(0.5,0.5)$) {};
	\node[vertex, label=below:{$w_2$}] (w2) at ($(v10)+(0.5,-0.5)$) {};
	\draw[edge,dashed] (u)--(u1);
	\draw[edge, dashed] (u)--(u2);
	\draw[edge, dashed] (w)--(w1);
	\draw[edge, dashed] (w)--(w2);
	\draw[edge] (u1)--(v01);
	\draw[edge] (v01)--(w1);
	\draw[edge] (u2)--(v00);
	\draw[edge] (v00)--(w2);
	\foreach \u/\v in {v00/v10,v10/v11,v11/v01,v01/v00}
	\draw[edge] (\u)--(\v);
\end{tikzpicture}
\caption{Illustrations for the vertex namings used in proofs, from left to right: \Cref{lemma:doublehouse}, \Cref{lemma:3-vertex-separator} and \Cref{lemma:contraction-connectivity}}\label{fig::prooffigs}
\end{figure}
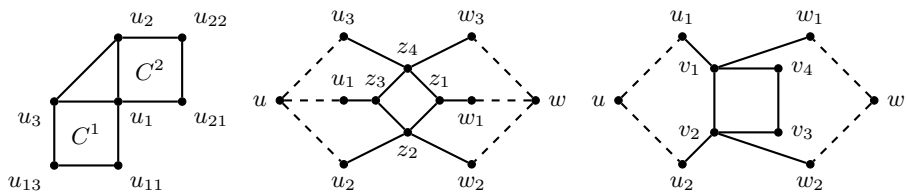

The following lemma captures the essence of why the algorithm provided in this section cannot be applied also on HHD-free graphs, since it holds for \ourclass graphs but not for HHD-free graphs.
Think for example of a simple path $P$ of length $5$ and a vertex disjoint induced chordless $C_4$, $C$. 
Consider also each vertex of $P$ being universal to $C$. Observe that this graph is HHD- but not \ourclass.
Every two non adjacent vertices of $C$ together with the endpoints of $P$ create an induced chordless $C_4$.
Adding a chord connecting the two endpoints of $P$ creates a hole and hence the resulting graph is not HHD-free.

\begin{lemma}\label{lemma:add-chord-to-cycle}
	Let $G$ be a \ourclass graph and $C=\{v_1,v_2,v_3,v_4\}$ an induced $C_4$ in $G$. Then the graph $G'$ created by adding the chord $v_1v_3$ to $G$ is  \ourclass and has one less induced $C_4$ than $G$.
\end{lemma}

An essential property of the graph class we work on is being closed under contraction, since our algorithm is based on contracting edges iteratively until the resulting graph becomes chordal.
Before proving this property though, although  ``after an edge contraction a new cycle is created'' is intuitively clear, we formally define what it means for a $C_4$ to be ``new''.
\begin{definition}
	Let $G$ be a graph, $uv\in E(G)$ and $G'=G/uv$. Let also $w$ be the vertex of $G'$ that is created by the contraction of $uv$. We say that an induced cycle $C$ containing $w$ in $G'$ is new if $N_C(w)\not\subseteq N_G(v)$ and $N_C(w)\not\subseteq N_G(u)$.
\end{definition}

\begin{lemma}\label{lemma:closed-under-contraction}
	\ourclass graphs are closed under contraction of an edge of an induced $C_4$.
\end{lemma}

In order to prove that the contractions of our algorithm do not affect the connectivity, we first study the possible role of vertices on an induced $C_4$ in minimal separators in \ourclass graphs.

\begin{lemma}\label{lemma:3-vertex-separator}
	Let $G$ be a $k$-connected \ourclass graph for $k\geq 5$.
	Then no three vertices of an induced $C_4$ belong in the same minimal separator.
\end{lemma}
\begin{proof}
	Let $G$ be a $k$-connected \ourclass graph for $k\geq 5$ and $v_1,v_2,v_3,v_4$ vertices that induce a $C_4$, $C$.
	Assume that $v_1,v_2,v_3$ belong in the a same minimal separator $S$ and hence, $(G\setminus\{v_1,v_2,v_3\})$ is only $k-3$ connected.
	Let also $u$ and $w$ be two distinct vertices belonging in different connected components of $G\setminus S$.
	
	Consider now the chordal graph $G'$ created, by adding $v_2v_4$ to $C$ and one chord to each other induced $C_4$ of $G$. By \Cref{lemma:add-chord-to-cycle} this is possible by adding exactly one chord to each induced $C_4$ of $G$ - in particular each addition does not create new induced $C_4$.
	Since $G'$ is chordal each minimal separator induces a clique, and hence $v_1,v_2,v_3$ cannot be part of the same minimal separator in $G'$ because they do not induce a triangle in $G'$. Thus $G'\setminus S$ remains connected.
	
	Let $P_1$ be a $u-w$ path  in $G'\setminus S$  that contains a minimal number of added edges. Let  $z_1z_3\in E(P_1)$ be one of the added edges, such that $z_3$ is closer to $u$ on $P_1$ than $z_1$.
	Note that 
	$z_1$ and $z_3$ are part of some induced $C_4$, $C'=\{z_1,z_2,z_3,z_4\}$ in $G$.
	Since $z_1z_3$ cannot be replaced by neither $z_1z_2,z_2z_3$, nor $z_1z_4,z_4z_3$ (otherwise we get a path with strictly less added edges than $P_1$) it follows that $z_2,z_4\in S$. 
	
	We will use the $u-w$ paths through $S$ in $G$  to reach a contradiction.
	Since $S$ is a minimal $u-w$ separator in $G$, there are two internally vertex disjoint  $u-w$ paths $P_2$ and $P_3$, with $P_2\cap S=\{z_2\}$ and $P_3\cap S=\{z_4\}$. 
	Let $w_2$ be the neighbor of $z_2$ on $P_2$ closer to $w$, $w_1$ the respective neighbor of $z_1$ on $P_1$ and $w_3$ the respective neighbor of $z_4$ on $P_3$. Let also $u_1$, $u_2$, $u_3$ be the corresponding neighbors of these paths closer to $u$. 
	See the illustration 
	in Figure~\ref{fig::prooffigs} for these namings, keeping in mind that it could be $w_1\in\{w_2,w_3\}$ or $u_1\in\{u_2,u_3\}$ or also $w_1=w_2=w_3=w$ or $u_1=u_2=u_3=u$.


We claim that, in $G$, $z_3$ is adjacent to a vertex on $P_2^{[w_2,w]}$ or $P_3^{[w_3,w]}$.
Assume otherwise, and assume that $P_2^{[w_2,w]}$, $P_3^{[w_3,w]}$ are induced paths in $G$ (shortcut them otherwise). If $w_2=w_3$ then notice that $w_2=w=w_3$. In order for $z_3,z_2,z_4,w$ not to induce a $C_4$  with three common vertices to $C$, $z_3$ has to be adjacent to $w$ which is on $P_2^{[w_2,w]}$. 
If $w_2\neq w_3$ then assume without loss of generality that $w_2\neq w$. In order to not be a hole, there has to be a chord in the cycle build by $P_2^{[w_2,w]}$, $P_3^{[w_3,w]}$ with $z_4,z_3,z_2$. By assumption, this chord cannot be from $z_3$, so it has to involve $z_4$ or $z_2$. Since $P_2^{[w_2,w]}$  and $P_3^{[w_3,w]}$ are induced and $w_2\neq w$, either $w_2$ is adjacent to $z_4$, or $w_3\neq w$ is adjacent to $z_2$. Both cases create a $C_4$ that has three vertices in common with $C$,  ($w_2,z_2,z_3,z_4$, and $w_3,z_2,z_3,z_4$, resp.) and since $z_4z_2\notin E(G)$, the added chord for these $C_4$ has to be $w_2z_3$, resp.~$w_3z_3$, leading again to $z_3$ being adjacent to some vertex on $P_2^{[w_2,w]}$ or $P_3^{[w_3,w]}$.

 Thus we conclude that $z_3$ is adjacent to a vertex $x$ on $P_2^{[w_2,w]}$ or $P_3^{[w_3,w]}$ in $G$. This however allows to create a path from $u$ to $w$ with (at least) one added edge less than $P_1$ in $G$ (since $P_2, P_3$ do not contain any added edges). Specifically, if $x$ is on $P_2$ we get  $P_1'=P_1^{[u,z_3]}xP_2^{[x,w]}$ and if $x\in P_3$,  $P_1'=P_1^{[u,z_3]}xP_3^{[x,w]}$.

Since $C'$ was an arbitrary cycle we conclude 
that $v_1,v_2,v_3$ cannot be part of the same minimal separator in $G$.
\end{proof}

\begin{lemma}\label{lemma:contraction-connectivity}
	Let $G$ be an \ourclass $k$-connected graph and $C=\{v_1,v_2,v_3,v_4\}$ be an induced $C_4$. The graph $G'=G/v_1v_2$ is still $k$-connected.
\end{lemma}

Now, we finally  look specifically at Algorithm~\ref{alg:: ourclassGL}, and first show that its subroutine creating $G''$ works correctly.
\begin{lemma}\label{lemma:existence of three vertices}
	Given an \ourclass graph $G$, the vertices selected in line 6 of Algorithm~\ref{alg:: ourclassGL} indeed exist as long as an induced $C_4$ exists.
\end{lemma}
\begin{proof}
Let $G$ be an \ourclass graph and $\mathcal C$ the set of all induced $C_4$ in $G$, consider the bipartite graph $T$ constructed through the following procedure:
	Its vertices are partitioned into two sets $B$, and $S$ referred to as \emph{big} and \emph{small} vertices of $T$, respectively.
	Each big vertex represents an induced $C_4$ of $\mathcal C$ while each small vertex represents a vertex of $G$ participating in at least two induced $C_4$.
	Each small vertex is adjacent to the big vertices which represent a $C_4$ this vertex participates in.
	We claim that with this definition $T$ is indeed a tree (actually a forest). Assume now for a contradiction that $T$ contains a cycle and let $C$ be one of the shortest such cycles in $T$. 
	
	First, consider the case that $C$ has length $l\geq 6$.
	Since $T$ is bipartite, due to its construction, $l$ is even and the vertices of $C=\{s_1,b_1,s_2,b_2,\ldots,s_{l/2},b_{l/2}\}$ alternate between big and small.
	We denote by $P_{b_i}^{s_w,s_z}$ a shortest path containing edges from the $C_4$ represented by the big vertex $b_i$ with endpoints the vertices represented by $s_w$ and $s_z$.
	Due to $C$, the cycle $P_{b_1}^{s_1,s_{l/2}}\ldots P_{b_l/2}^{s_{l/2-1},s_{l/2}}$ exists in $G$ as a subgraph. Note that since we have assumed that $C$ is a minimal length cycle of $T$ it is also chordless.
	Hence, in order for a hole not to be an induced subgraph of $G$ at least one chord must exist connecting two vertices corresponding to two small ones of $T$.
	This however would create a double house as a subgraph with the two $C_4$ forming it being the two that correspond to big vertices of $T$. By  Lemma~\ref{lemma:doublehouse} this means that one of the  $C_4$ is not induced, a contradiction to the construction of $T$.
	Notice also that in the case where $l=6$ we  directly find a double house and reach a contradiction using the same arguments.
	
	Moreover the assumption that $l=4$, leads us to a contradiction to the fact that two $C_4$ have at most one vertex in common.
	Hence, $T$ is a forest and the vertices mentioned in line 6 are the ones belonging only to a cycle represented by one leaf belonging in $B$.
\end{proof}

\begin{lemma}\label{lemma:C4-tree}
	Given an \ourclass graph $G$, lines 1-10 of Algorithm~\ref{alg:: ourclassGL} transforms $G$ into a weighted chordal graph $G''$, with the same connectivity as $G$ and such that each vertex from $G$ is involved in at most one edge contraction to create $G''$.
\end{lemma}

At last, notice that we can easily alter Algorithm~\ref{alg:: ourclassGL} to also work for weighted graphs, with the simple change of setting the weights of a vertex in $G''$ in line 10 to the sum of the weights of the original vertices it was contracted from.
With this alteration, we can conclude now the proof of \Cref{thm:ourclass-GL-partition} with the following.

\begin{lemma}
	\label{lemma::running-time}
 Algorithm~\ref{alg:: ourclassGL} works correctly and runs in time $\bigO(|V|^4)$.	
\end{lemma}
\begin{proof}
	By \Cref{lemma:C4-tree}, $G''$ is a chordal graph with maximum vertex weight $2\wmax$. Further, observe that we did not merge terminal vertices with each other, thus we can properly run Algorithm~\ref{alg::chordal-weighted} on it. By the correctness of this algorithm (\Cref{thm::GL-theorem-chordal-weighted}), we know that $S_1,\dots,S_k$  in line 11 is a GL-partition for $G''$ with deviation $2\wmax$. Since reversing edge-contraction does not disconnect these sets, the unfolded sets $S_1,\dots,S_k$ are thus also a GL-partition for $G$ with deviation $2\wmax$; note here that the only edges we added to create $G''$ are between terminal vertices, which are in separate sets $S_i$ by definition.
	
	The most time consuming part of  Algorithm~\ref{alg:: ourclassGL} is the preprocessing to transform the input graph into a weighted chordal graph which requires $\mathcal \bigO({|V|}^4)$ time in order to find all the  induced $C_4$ (note that the induced $C_4$ are at most $(n-4)/3$ since they induce a tree).
\end{proof}

Moreover, as is the case for chordal graphs, we can sacrifice terminals to enlarge the considered graph class.
\begin{corollary}
	Let $G$ be a $k$-connected graph which becomes \ourclass after adding $k/2$ edges. Then, given those edges, a \glp of $G$ with deviation $1$ (also its weighted version with deviation $2w_{max}-1$) can be computed in polynomial time but without fixed terminals.
\end{corollary}


%
%
%
 \bibliographystyle{splncs04}
%
\bibliography{chordalGL} 
\newpage
\appendix

\section{Correctness of Algorithm~\ref{alg::chordal-weighted} (\algChordalGLweighted)}\label{appendix1}

For the correctness Algorithm~\ref{alg::chordal-weighted}, we need to prove that all vertices are assigned, i.e.~the algorithm terminates, and if this is the case, then $S_1, \dots, S_k$ corresponds to a connected vertex partition satisfying the required weight conditions for each $S_i$.
We start by proving that the Algorithm~\ref{alg::chordal-weighted} eventually ends up assigning all vertices into connected vertex sets.

\begin{lemma}
	\label{lemma::weighted-terminates-conncected-partition}
	Algorithm~\ref{alg::chordal-weighted} terminates, such that all vertices are assigned, where $G[S_i]$ is connected for all $i \in [k]$.
	Further, the while-loop iterates at most $|V|$ times.
\end{lemma}
\begin{proof}
	Similarly to Algorithm~\ref{alg::chordal}, because we add only vertices to non-full weighted sets from its unassigned neighborhood, the $S_i$s correspond always to connected vertex sets.
	Note that by \Cref{corollary::never_iso} as long as we have non-full weighted sets and unassigned vertices, the while-loop makes progress in the sense that either an unassigned vertex becomes assigned or a set is declared to be a full weighted set.
	Thus, if there are unassigned vertices after the while-loop, and we have $|I| = 1$ and we assign the remaining ones to the last non-full weighted set $S_j$ with $j \in I$ (cf.~line~\ref{algWeightedChordal::last-step}).
	Observe that \Cref{corollary::never_iso} implies that $G[S_j]$ is still connected after adding the remaining vertices.
	
	For the second part of the lemma, when we reach the while loop, there are exactly $|V| - k$ unassigned vertices.
	Except of at most $k$ times an unassigned vertex becomes assigned in an iteration of the while-loop.
	This in turn implies that we have not more than $|V|$ iterations.
\end{proof}

The running time in \Cref{thm::GL-theorem-chordal-weighted} is $\bigO(|V|^2)$ since the while-loop iterates at most $|V|$ times and each operation in this loop runs in $\bigO(|V|)$ time.
Hence, to prove \Cref{thm::GL-theorem-chordal-weighted} it remains to show that the required weight condition for each part of the connected vertex partition $S_1, \dots, S_k$ is satisfied.

The indices in $I$ denote the non-full weighted sets and therefore, $\compi := [k] \setminus I$ the indices of the full weighted sets.
Declaring a set $S_{j'}$ as complete weighted set, i.e.~we remove $j'$ from $I$, implies that $w(S_{j'}) + w(v') \geq w_{j'}$.
If $w(S_{j'}) + w(v') \neq w_{j'}$, whether we add $v'$ to $S_{j'}$ depends on whether the value of $\sum_{i \in \compi} (w_i - w(S_i))$ is less than zero or not.
Basically, this sum serves to balance the variations in the required weights of the sets $S_1, \dots, S_k$, and determines the moment we declare $S_{j'}$ to be a full weighted set whether we want $w(S_{j'}) < w_{j'}$ or $w(S_{j'}) > w_{j'}$.
During the algorithm the sum $\sum_{i \in \compi} (w_i - w(S_i))$ satisfies the following invariant. 

\begin{lemma}
	\label{lemma::invariant-weights}
	In the Algorithm~\ref{alg::chordal-weighted}, before reaching line~\ref{algWeightedChordal::last-step} we have each time $|\sum_{i \in \compi} (w_i - w(S_i))| < \wmax$.
\end{lemma}
\begin{proof}
	We prove this lemma by induction on $|\compi|$.
	After assigning each terminal to a corresponding vertex set, we initialize $I$ by $I = \{i \in [k] \mid w(S_i) < w_i\}$.
	That is, $\sum_{i \in \compi} (w_i - w(S_i)) = 0$, since either $\compi = \varnothing$ or each $i \in \compi$ satisfy $w(S_i) = w_i$ by $w(t_i) \leq w_i$.
	
	Assume $|\sum_{i \in \compi} (w_i - w(S_i))| < \wmax$ for $|\compi| \leq \ell < k$ and we now add $j'$ to $\compi$ according to the algorithm, i.e.~we remove $j'$ from $I$.
	First, we show that $|w_j - w(S_{j'})| < \wmax$ in both possible future cases $v' \in S_{j'}$ or $v' \notin S_{j'}$.
	By $w(S_j \cup \{v'\}) = w_j$ we have added $v'$ to $S_{j'}$ and $|w_j - w(S_{j'})| = 0 < \wmax$ holds.
	Thus, we can assume that $w(S_{j'} \setminus \{v'\}) < w_j$ and $w(S_{j'} \cup \{v'\}) > w_{j'}$ which in turn results to $|w_{j'} - w(S_{j'} \cup \{v'\})| < \wmax$ and $|w_{j'} - w(S_{j'} \setminus \{v'\})| < \wmax$ by $w(v') \leq \wmax$.
	
	If $w(S_{j'} \cup \{v'\}) = 0$, we have $\sum_{i \in \compi \setminus \{j'\}} (w_i - w(S_i)) = \sum_{i \in \compi} (w_i - w(S_i))$
	and we are done by the induction hypotheses
	
	We can assume that $w(S_{j'} \cup \{v'\}) > w_{j'}$.
	If $0 \leq \sum_{i \in \compi \setminus \{j'\}} (w_i - w(S_i)) < \wmax$ we add $v'$ to $S_{j'}$.
	It follows that $\sum_{i \in \compi} (w_i - w(S_i)) < \sum_{i \in \compi \setminus \{j'\}} (w_i - w(S_i))$ by $w_{j'} - w(S_{j'}) < 0$.
	By $-\wmax < w_j - w(S_{j'}) < 0$ the sums might deviate by at most $\wmax-1$ from each other.
	Thus, by $\sum_{i \in \compi \setminus \{j'\}} (w_i - w(S_i)) \geq 0$ we obtain $|\sum_{i \in \compi} (w_i - w(S_i))| < \wmax$.
	
	
	In case $-\wmax < \sum_{i \in \compi \setminus \{j'\}} (w_i - w(S_i)) < 0$ we do not add $v'$ to $S_{j'}$ and obtain $\sum_{i \in \compi} (w_i - w(S_i)) > \sum_{i \in \compi \setminus \{j'\}} (w_i - w(S_i))$ by $w_{j'} - w(S_{j'}) > 0$.
	Furthermore, by $0 < w_j - w(S_{j'}) < \wmax$ the sums deviate by at most $\wmax-1$ from each other and finally, by $\sum_{i \in \compi \setminus \{j'\}} (w_i - w(S_i)) < 0$ we obtain $|\sum_{i \in \compi} (w_i - w(S_i))| < \wmax$.
	%
	
\end{proof}

With \Cref{lemma::invariant-weights} we prove now the last part of the proof \Cref{thm::GL-theorem-chordal-weighted}.

\begin{lemma}
	If Algorithm~\ref{alg::chordal-weighted} terminates, then we have $w_i - \wmax < w(S_i) < w_i + \wmax$ for each $i \in [k]$. 
\end{lemma}
\begin{proof}
	The sets in $S_1, \dots, S_k$ with indices $[k] \setminus I$ in the initialization of $I$, i.e.~$I = \{i \in [k] \mid w(S_i) < w_i\}$, satisfy clearly its weight conditions as $w(t_i) \leq w_i$ for all $i \in [k]$.
	Next, we show that each set that is declared as full weighted set in the while-loop
	satisfies its weight condition, i.e.~$|w_i - w(S_i)| < \wmax$ for $i \in \compi$.
	According to Algorithm~\ref{alg::chordal-weighted} let $j'$ be the index that we remove from $I$ and consider $S_{j'}$ before possibly adding $v'$ to it.  
	$w(S_{j'}) < w_{j'}$ and $w(S_{j'}) + w(v') \geq w_{j'}$ implies that $w(S_{j'}) < w_{j'} + \wmax$ independent of $v'$ being added to $S_{j'}$ or not by $w(v') \leq \wmax$.
	Similarly, $w(S_{j'}) + w(v') > w_{j'}$ implies that $w(S_{j'}) > w_{j'} - \wmax$.
	In case $w(S_{j'}) = w_{j'} - \wmax$ and $w(v') = \wmax$ the algorithm adds $v'$ to $S_{j'}$ (cf.~line~\ref{alg::weighted-if2}) and we have $w_{j'} - w(S_{j'} \cup \{v'\}) = 0 < \wmax$.
	
	Hence, it remains to prove that the weight conditions are satisfied from the non-full sets if either the while-loop terminates with all vertices assigned, or with $|I| = 1$.
	We start with the former case.
	If all vertices are assigned we have $\sum_{i=1}^{k}(w_i - w(S_i)) = \sum_{i=1}^{k} w_i - \sum_{i=1}^{k} w(S_i)) = w(V) - w(V) = 0$.
	Let $I$ be the indices of the non-full weighted sets after the while-loop is terminated with all vertices assigned and recall $\compi = [k] \setminus I$.
	Each non-full weighted set $S_i$ for $i \in I$ satisfies $w(S_i) < w_i$ and therefore $\sum_{i \in I} w_i - w(S_i) > 0$ as each value $w_i - w(S_i)$ is greater than zero.
	Thus, by $\sum_{i=1}^{k}(w_i - w(S_i)) = \sum_{i \in \compi} (w_i - w(S_i)) + \sum_{i \in I} (w_i - w(S_i)) = 0$ we have $\sum_{i \in \compi} (w_i - w(S_i)) < 0$ and by \Cref{lemma::invariant-weights} $-\wmax < \sum_{i \in \compi} (w_i - w(S_i)) < 0$.
	Suppose there is an $i \in I$ with $w_i - w(S_i) \geq \wmax$.  
	This would imply that $0 = \sum_{i \in \compi} (w_i - w(S_i)) + \sum_{i \in I} (w_i - w(S_i)) > -\wmax + \sum_{i \in I} (w_i - w(S_i)) \geq -\wmax + \wmax = 0$, which is a contradiction.
	
	It remains to consider the case that the while-loop terminates when $|I| = 1$.
	Let say $I = \{\ell\}$ and the remaining unassigned vertices are already added to $S_\ell$ according to line \ref{algWeightedChordal::last-step}.
	Same as above, the $S_i$'s with $i \in [k] \setminus \{\ell\} = \compi$ satisfy its required weight condition and hence, we need to show that $w_\ell - \wmax < w(S_\ell) < w_\ell + \wmax$ holds.
	By $\sum_{i \in \compi} (w_i - w(S_i)) + \sum_{i \in I} (w_i - w(S_i)) = \sum_{i \in \compi} (w_i - w(S_i)) + (w_\ell - w(S_\ell)) = 0$ we obtain $\sum_{i \in \compi} (w_i - w(S_i) = w(S_\ell) - w_\ell$.
	As a result, since $|\sum_{i \in \compi} (w_i - w(S_i)| < \wmax$ by \Cref{lemma::invariant-weights}, it follows that $|w(S_\ell) - w_\ell| < \wmax$.
\end{proof}

\section{Omitted Proofs}\label{appendix2}
\subsection*{Proof of Lemma~\ref{lemma::ordering_induced_path}}
\begin{proof}
	Assume for a contradiction that there is an index $j\in\{2,\ldots,k-1\}$ such that $\sigma(v_j)=\min_{i\in[k]}\{\sigma(v_i)\}$.
	As a result $\sigma(v_{j-1})>\sigma(v_j)$ and $\sigma(v_j)<\sigma(v_{j+1})$.
	This means however when $v_j$ is assigned a $\sigma$ value, none of $v_{j-1}$ and $v_{j+1}$ have been assigned such a value. 
	Since $v_j$ is a simplicial vertex at that point, its neighbors that still have not been assigned a $\sigma$ value induce a clique.
	Hence, $v_{j-1}v_{j+1}\in E(G)$, which contradicts the fact that $P$ induces a simple path.
\end{proof}


\subsection*{Proof of Lemma~\ref{lemma:observation}}
\begin{proof}
	Let $v\in V\setminus C$  be adjacent to two vertices $u$ and $w$ in $C$. 
	If $uw\in E(G)$ then, since $G$ is house-free, $v$ is also adjacent to at least one vertex of $C$ other than $u$ and $w$. 
	Then however, if $v$ is only adjacent to three vertices of $C$, it induces a $C_4$ with two of the adjacent and the non adjacent vertex, that shares 3 vertices with $C$.
	In the case where $uw\not\in E(G)$, in order to not have two induced $C_4$ sharing three vertices, $v$ has to also be adjacent to another vertex of $C$, which leads us to the previous case where $v$ is adjacent to two vertices of $C$ inducing an edge and hence universal to $C$.
	
	Moreover notice that any set of universal vertices to $C$, induces a clique since otherwise two induced $C_4$ exist that share two vertices (consider $C$ and the $C_4$ induced by two non adjacent vertices that are universal to $C$  and two non adjacent vertices of $C$).
\end{proof}

\subsection*{Proof of Lemma~\ref{lemma:doublehouse}}
\begin{proof}
	Let $G$ be  an \ourclass graph and $H$ be a subgraph of $G$ that is a double house.	As illustrated in Figure~\ref{fig::prooffigs}, we denote by $C^1$ and $C^2$ the two induced $C_4$ of $H$, by $u_1$ their common vertex, by $u_2$ and $u_3$ the ones adjacent to $u_1$ and to each other, belonging in $C^2$ and $C^1$, respectively and finally by $u_{ij}$ the vertex of $C_i$ that is adjacent to $u_j$.
	
	Notice that by Lemma~\ref{lemma:observation}, since $u_2$ is adjacent to two vertices of $C^1$ it is universal to $C^1$, while the same holds for $u_3$ and $C^2$.
	After adding those edges to $H$ however $u_{21}$ is adjacent to two of the vertices of $C^1$, and hence is universal to $C^1$ in $G$, while the same holds for $u_{11}$ and $C^2$.
	After this addition however $u_{13}$ is adjacent to both $u_2$ and $u_{21}$ of $C^2$ and hence universal to $C^2$, which creates the chord $u_{13}u_1$ in $C^1$ that concludes this proof.
\end{proof}

\subsection*{Proof of Lemma~\ref{lemma:add-chord-to-cycle}}
\begin{proof}
	Assume for a contradiction that the addition of $v_1v_3$ creates a new induced $C_4$, $C'=\{v_1',v_1,v_3,v_3'\}$, where $v_1'$, and $v_3'$ are the neighbors on $C'$ of $v_1$ and $v_3$ respectively.
	Notice that $v_1'$ and $v_3'$ do not belong in $C$.
	In order for $v_1,v_2,v_3,v_3',v_1'$ not to induce a hole in $G$, either $v_1'$ or $v_3'$, say $v_1'$, has to also be adjacent to either $v_2$ or $v_3$. 
	This however, due to \Cref{lemma:observation}, means that $v_1'$ is universal to $C$ in $G$, which creates a chord on $C'$.
	
	It remains now to show that $G'$ is still $\text{HHI}_4^2$-free.
	Using similar arguments as before we see that no hole is formed from the addition of $v_1v_3$ and since no new $C_4$ is formed also, the remaining $C_4$ keep having pairwise at most one vertex in common.
	Assume now for a contradiction that adding  $v_1v_3$  creates an induced house $H$ in $G'$.
	Since, as we showed above, $v_1v_3$ does not participate in any induced $C_4$ it must be one of the two roof's edges.
	Notice also that from the vertices of $C$, only $v_1$ and $v_3$ participate in this house because otherwise $G$ would contain two induced $C_4$ sharing two vertices.
	Let $u$ be the third vertex of the roof and notice that by \Cref{lemma:observation}, $u$ is also adjacent to $v_2$ and $v_4$, and let $w$ and $z$ be the remaining two vertices of the house (assume that $wu, zv_3\in E(G)$).
	In order for $wzv_2v_3u$ not to induce a house in $G$ either $v_2z\in E(G)$ or $v_2w\in E(G)$.
	Notice that through these cases we conclude, again by \Cref{lemma:observation},  that either $w$ or $z$ is universal to $C$. 
	We have assumed however that $H$ is an induced house, hence $w$ is not universal to $C$ because that would create a chord in the body of the house.
	As a result $z$ is universal to $C$, which again leads to a contradiction to the fact that $H$ is an induced house because of the edge $v_1z$.
\end{proof}

\subsection*{Proof of Lemma~\ref{lemma:closed-under-contraction}}
\begin{proof}
	As we have stated before \ourclass graphs are in particular HHD-free. Since HHD-free graphs are closed under edge contraction, no hole or house occurs after contracting any edge of an \ourclass graph.	
	
	Let $C=\{v_1,v_2,v_3,v_4\}$ be an induced $C_4$ in $G$ and consider contracting the edge  $v_1v_2$. Let $G'$ be the graph resulting from this contraction, and let $v_{12}$ be the newly added vertex.
	
	We first show that contracting  $v_1v_2$ does not create any new $C_4$.
	Assume for a contradiction that $G'$ contains a new induced $C_4$, $C'=\{v_{12},u_1,u_2,u_3\}$.
	\begin{itemize}
		\item
		$v_3\not\in C'$
		
		Let $u_1,u_2$ be the neighbors of $v_{12}$ on $C'$.
		Since $C'$ is new we have that $u_1v_1,u_2v_2\in E(G)$ and $u_1v_2,u_2v_1\not\in E(G)$.
		In order for $u_3,u_2,v_2,v_1,u_1$ to not induce a hole in $G$, at least one of $u_3v_1, u_3v_2, u_1u_2$ has to also exists as edges in $G$.
		This however would create a chord in $C'$  which contradicts our assumption that $C'$ is an induced $C_4$ in $G'$.
		\item
		$v_3\in C'$
		
		In order for $v_{12}v_3$ not to be a chord in $C'$, this edge participates in the induced $C_4$, and also $u_1,u_3\neq v_2,v_4$ (assuming $v_3=u_2$ and that $u_3v_3,u_1v_{12}\in E(G')$).
		Since $C'$ is new and $v_2v_3\in E(G)$ we conclude that $v_1u_1\in E(G)$ and $v_2u_1\not\in E(G)$.
		Notice now that in order for $v_1,u_1,v_2,v_3,u_3$ not to induce a hole in $G$, $u_3$ should be adjacent in $G$ to either $v_1$ or $v_2$.
		This however would create a chord on $C'$ in $G'$ which leads to a contradiction.
	\end{itemize}
	
	Now it remains to show that contracting $v_1v_2$ does not create two induced $C_4$ that share more than one vertex. 
	Assume now for a contradiction that $G'$ contains two 
	$C_4$ $Z=\{z_1,z_2,z_3,z_4\}$, $W=\{w_1,w_2,w_3,w_4\}$ that share more than one vertex. 
	
	First, assume that $Z$ and $W$ share an edge, thus let $z_1=w_1$ and $z_2=w_2$ be two of the common vertices of $Z$ and $W$. We can directly assume that $z_1=v_{12}$, and $z_3,z_4\notin W$ and $w_3,w_4\notin Z$, since previously any pair of $C_4$ shared at most one vertex. 
	By \Cref{lemma:add-chord-to-cycle}, contracting $v_1v_2$ did not create any new induced $C_4$, thus $Z$ and $W$ with the vertex $v_{12}$ replaced by either $v_1$ or $v_2$ were already induced $C_4$ in $G$, so assume that $\{v_1,w_2,w_3,w_4\}$ induces a $C_4$ in $G$. 
	
	Since $G$ was \ourclass it follows that $v_1$ is not adjacent to $z_4$ (otherwise $\{v_1,w_2,w_3,w_4\}$ and $\{v_1,z_2,z_3,z_4\}$ are two induced $C_4$ sharing more than one vertex), thus  $\{v_2,z_2,z_3,z_4\}$ also induces a $C_4$ in $G$.
	
	Then, however in order for $v_1,v_2,z_2,w_3,w_4$ not to induce a house in $G$ either $w_3$ or $w_4$ is adjacent to $v_2$, which would create a chord in $W$ after the contraction of $v_1v_2$ that leads to a contradiction.
	
	It remains to consider the case that $Z$ and $W$ share two non-adjacent vertices, i.e.~$w_1=z_1=v_{12}$ and $w_3=z_3$ are the two common vertices. Similarly to the previous case, we can use Lemma~\ref{lemma:add-chord-to-cycle} to assume that $\{v_1,w_2,w_3,w_4\}$ and $\{v_2,z_2,z_3,z_4\}$ are induced $C_4$ in $G$. Since  $v_1,v_2,z_2,z_3$ does not create an induced house or hole together with $w_2$ or $w_4$, it follows that either $v_1$ is adjacent to $z_3$, which would create a chord for $Z$ in $G'$, or  $w_2$ and $w_4$ are adjacent to $v_1$ or $z_2$. In the latter case, $w_2$ and $w_4$ are adjacent to more than one vertex of the induced $C_4$ $Z$, which means they are both universal to $Z$ and have to form a clique by Lemma~\ref{lemma:observation}. Then however $w_2w_4$ is a chord for $W$ in $G'$. 
\end{proof}

\subsection*{Proof of Lemma~\ref{lemma:contraction-connectivity}}

\begin{proof}
	Assume for a contradiction that $G'$ is only $k-1$-connected, thus there is a separator of size $k-1$ that disconnects two distinct vertices $u$ and $w$ in $G'$.
	
	Since the connectivity between  $u$ and $w$  dropped  after contracting $v_1v_2$, in $G$ there are two internally vertex disjoint paths $P_1$ and $P_2$ that connect $u$ and $w$ such that $v_1\in P_1$ and $v_2\in P_2$.
	We can further assume that $P_1$ and $P_2$ are two such paths of minimal length, respectively.
	Denote also by $u_1$ and by $w_1$ the neighbors of $v_1$ on $P_1$ which are closer to $u$ and $w$, respectively. (See the right illustration in Figure~\ref{fig::prooffigs} for an example of these namings, keeping in mind that it could be that $u_1=u_2=u$ and/or $w_1=w_2=w$.)
	Similarly, denote by $u_2$ and $w_2$ these neighbors of $v_2$ on $P_2$.
	
	In order for $v_1$, $v_2$ and $u$ not to be part of an induced hole (given that $u_1$ and $u_2$ are not $u$) either $u_1u_2\in E(G)$ or $u_1v_2\in E(G)$ or $u_2v_1\in E(G)$.
	The first case, however, if no other edges existed, would create a domino while the other cases form a house.
	Hence at least one of the edges $u_1v_3$, $u_1v_4$, $u_2v_3$, $u_2v_4$ exists.
	Similarly, assuming that neither $w_1$ nor $w_2$ is equal to $w$, since our graph is \ourclass, also one of $w_1v_3$, $w_1v_4$, $w_2v_3$, $w_2v_4$ exists.
	We show now that by using the remaining vertices of $C$ we can recreate the two vertex disjoint paths that previously existed in $G$.
	\begin{enumerate}
		\item $u_1v_3\in E(G)$
		\begin{enumerate}
			\item $w_1v_3\in E(G)$:
			Notice that the two $u-w$ vertex disjoint paths are preserved in $G'$, specifically $P_1'=P_1^{[u,u_1]}v_3P_1^{[w_1,w]}$ and $P_2'=P_2$ (where in $P_2$ we have instead of $v_2$ the newly created vertex $v_{12}$).
			
			\item $w_1v_4\in E(G)$:
			$P_1'=P_1^{[u,u_1]}v_3v_4P_1^{[w_1,w]}$, $P_2'=P_2$
			
			\item $w_2v_3\in E(G)$:
			$P_1'=P_1^{[u,u_1]}v_3w_2P_2^{[w_2,w]}$, $P_2'=P_2^{[u,v_{12}]}P_1^{[v_{12},w]}$
			
			\item $w_2v_4\in E(G)$:
			$P_1'=P_1^{[u,u_1]}v_3v_4w_2P_2^{[w_2,w]}$, $P_2'=P_2^{[u,v_{12}]}P_1^{[v_{12},w]}$
		\end{enumerate}
		\item $u_1v_4\in E(G)$
		\begin{enumerate}
			\item $w_1v_3\in E(G)$:
			$P_1'=P_1^{[u,u_1]}v_4v_3P_1^{[w_1,w]}$ and $P_2'=P_2$ 
			
			\item $w_1v_4\in E(G)$:
			$P_1'=P_1^{[u,u_1]}v_4P_1^{[w_1,w]}$, $P_2'=P_2$
			
			\item $w_2v_3\in E(G)$:
			$P_1'=P_1^{[u,u_1]}v_4v_3w_2P_2^{[w_2,w]}$, $P_2'=P_2^{[u,v_{12}]}P_1^{[v_{12},w]}$
			
			\item $w_2v_4\in E(G)$:
			$P_1'=P_1^{[u,u_1]}v_4w_2P_2^{[w_2,w]}$, $P_2'=P_2^{[u,v_{12}]}P_1^{[v_{12},w]}$
		\end{enumerate}
		\item The remaining cases are symmetrical to the ones written above.
	\end{enumerate}
	Similar arguments can  be used to find such paths if $u_1=u_2=u$ or $w_1=w_2=w$.
	
	Due to \Cref{lemma:3-vertex-separator} we know that $v_3$ and $v_4$ are free to be used for the creation of the above paths since they can not be part of the same minimal separator as $v_1$ and $v_2$.
\end{proof}

\subsection*{Proof of Lemma~\ref{lemma:C4-tree}}
\begin{proof}
	Observe that by \Cref{lemma:add-chord-to-cycle} we can safely add the edges in line 1, in the sense that we still have an \ourclass graph, and that we do not create new $C_4$ that our terminal vertices might participate in. Thus, the graph we consider moving forward in the algorithm is \ourclass and the induced $C_4$ considered in line 2 do not contain two non-adjacent terminals. 
	
	Consider now the tree $T$ constructed by $\mathcal C$ in \Cref{lemma:existence of three vertices}. We proceed in arguing that the contraction order described in line 6 of Algorithm~\ref{alg:: ourclassGL} is indeed the desired one.
	
	
	Let $v$ be a leaf of $T$. If $v\in S$ (representing a single vertex in $G$) then delete $v$ and update $T$ accordingly.
	If $v\in B$ then consider the induced $C_4$ corresponding to $v$. 
	The fact that $v$ is a leaf in $T$ means that there are at least three vertices in the corresponding $C_4$ that are not included in any other induced $C_4$ of $G$. These three vertices $v_1,v_2,v_3$ are candidates for line 6 of Algorithm~\ref{alg:: ourclassGL}, and after removal of them and the  $C_4$ corresponding to $v$ in $G'$ in line~8, removing $v$ from $T$ yields a tree for which this procedure can be repeated.
	
	Since we remove all vertices involved in a contracted edge from $G'$ as soon as their first adjacent edge is chosen, we can ensure that no vertex is contracted twice. Thus when we create $G''$ all vertices indeed have maximum weight at most~2.
\end{proof}






\end{document}